\documentclass[11pt,fleqn]{article}
\usepackage{amsmath,amsthm,amsfonts,amssymb,mathrsfs,environ}
\usepackage{accents}
\usepackage{geometry}
\usepackage{bm}
\usepackage{bbm}
\usepackage{dsfont}
\usepackage{calrsfs}
\usepackage{rotating}
\usepackage{hyperref}
\usepackage{longtable}
\usepackage{url}
\usepackage{enumerate}
\usepackage[round,longnamesfirst]{natbib}
\usepackage{setspace}
\usepackage[tableposition=top,singlelinecheck=off,labelfont=bf]{caption}
\usepackage{tabularx}
\usepackage{graphicx}
\usepackage{subfigure}
\usepackage{verbatim}
\usepackage{color}
\usepackage{pdflscape}
\usepackage{afterpage}
\usepackage{caption}
\usepackage{breqn}
\usepackage{appendix}
\usepackage{caption}
\DeclareMathOperator{\E}{\mathbb{E}}

\DeclareMathOperator{\Prob}{\mathbb{P}}

\DeclareMathOperator{\diag}{\text{diag}}

\newcommand{\abs}[1]{\left\lvert#1\right\rvert}
\newcommand{\norm}[1]{\left\lVert#1\right\rVert}

\newcommand{\bE}{\bm E}

\newcommand{\bI}{\bm I}

\newcommand{\bS}{\bm S}
\newcommand{\bs}{\bm s}

\newcommand{\bU}{\bm U}

\newcommand{\bV}{\bm V}
\newcommand{\bv}{\bm v}

\newcommand{\bX}{\bm X}
\newcommand{\bx}{\bm x}
\newcommand{\bY}{\bm Y}
\newcommand{\by}{\bm y}

\newcommand{\bbeta}{\bm \beta}

\newcommand{\bepsilon}{\bm \epsilon}

\newcommand{\bLambda}{\bm \varLambda}

\newcommand{\bSigma}{\bm \varSigma}

\theoremstyle{definition}

\newtheorem{assumption}{Assumption}

\newtheorem{remark}{Remark}
\theoremstyle{plain}
\newtheorem{theorem}{Theorem}
\newtheorem{lemma}{Lemma}
\newtheorem{corollary}{Corollary}

\title{A restricted eigenvalue condition for unit-root non-stationary data}
\author{Etienne Wijler}
\date{Vrije Universiteit Amsterdam\\
Department of Econometrics and Data Science\\
\today}

\begin{document}
\maketitle

\begin{abstract}
    In this paper, we develop a restricted eigenvalue condition for unit-root non-stationary data and derive its validity under the assumption of independent Gaussian innovations that may be contemporaneously correlated. The method of proof relies on matrix concentration inequalities and offers sufficient flexibility to enable extensions of our results to alternative time series settings. As an application of this result, we show the consistency of the lasso estimator on ultra high-dimensional cointegrated data in which the number of integrated regressors may grow exponentially in relation to the sample size.\\


\textit{JEL-Codes}: C32, C55
\end{abstract}

\section{Introduction}

Modern time series applications frequently involve the analysis of high-dimensional datasets in which the number of time series $N$ is large in relation to, and may exceed, the number of temporal observations $T$. A rather successful approach to modelling such datasets has been to impose some form of sparsity on the underlying data generating process (DGP) and to incorporate $\ell_1$-based regularization such as the lasso to reduce model complexity and/or recover the true sparsity pattern \citep[e.g][]{Basu2015,Medeiros2016,Masini2019}. However, these approaches heavily rely on stationarity of the considered time series, thereby either excluding the presence of stochastic trends or requiring correct identification and corrections of the order of integration. As noted by \citet{Smeekes2020a}, a key difficulty that hinders extensions of the theory for lasso-type estimators to high-dimensional non-stationary settings is the absence of a verifiable restricted eigenvalue condition (REC). Accordingly, we aim to fill this gap in the literature by proposing an appropriately scaled REC for unit root non-stationary time series and verifying its validity directly on the gram matrix with the use of matrix concentration inequalities. This result is then used to derive novel finite sample error-bounds for the lasso and prove its asymptotic consistency on (ultra) high-dimensional (co)integrated datasets.

The behaviour of $\ell_1$-based regularization on stationary time series data is well-researched. \citet{Basu2015} and \citet{KockCallot2015} consider lasso based estimation of high-dimensional VAR processes driven by Gaussian innovations and derive tight error bounds for the prediction and estimation error. These bounds show that the estimation error deteriorates in the dimension at a logarithmic rate, thereby validating the lasso in (ultra) high-dimensional time series settings. \citet{Medeiros2016} consider the adaptive lasso and derive both estimation and selection consistency on weakly dependent and possibly heteroskedastic time series data without imposing Gaussianity on the predictors or innovations. By replacing the assumption of Gaussianity by a set of moment restrictions, they provide a considerable generalization to previous results, albeit at the cost of error bounds that deteriorate at faster rates in the dimension. \citet{Masini2019} generalize the aforementioned results on sparse VARs estimated with the lasso by deriving error bounds under the assumption of weak sparsity, while only assuming the innovation vector to be geometrically strong ($\alpha$-)mixing and a martingale difference process. Finally, \citet{Wong2020} provide further generalizations by avoiding the assumption of a specific parametric form of the DGP and relying only on (strict) stationarity and geometrically decaying $\beta$-mixing coefficients to establish error bounds for the lasso for subweibull random vectors.

Owing to the lack of suitable restricted eigenvalue conditions, the literature on the lasso in the non-stationary setting is substantially more scarce, with current results being limited to fixed or moderate dimensional settings only. For example, in the fixed dimensional setting, \citet{Liao2015} develop a penalized estimator for a vector error-correction model (VECM) and derive its consistency and pointwise limit distribution. \citet{Lee2021} consider a two-stage adaptively weighted lasso procedure applied to (co)integrated data for which the oracle property is shown to hold. \citet{Koo2020} consider the lasso applied to predictive regressions with covariates that are integrated of order one and derive error bounds under the assumption that the gram matrix satisfies an REC. The authors correctly note that the resulting lower bound that is assumed on the restricted eigenvalue may decrease to zero as the number of integrated time series grows, thereby inflating the estimation error at an unknown rate. Accordingly, consistency is again only derived for the case of a fixed number of integrated time series. Moving on to the moderate dimensional setting, \citet{Liang2019} develop a sparse VECM estimator that is argued to be consistent in a setting where the number of stochastic trends in the data increases at a rate slightly slower than $T^{1/4}$. \citet{Smeekes2020a} consider lasso-type estimation of a single-equation error-correction model and derive consistency in an asymptotic framework in which the number of integrated time series is allowed to grow at rate $T^{1/4}$. In addition, the authors show that the minimum eigenvalue of the usual non-stationary gram matrix converges to zero as the dimension increases, but can be bounded away from zero with probability converging to one when scaled up by its dimension if the latter grows no faster than rate $T^{1/2}$.

To the best of the author's knowledge, this paper is the first to develop and verify a restricted eigenvalue condition that enables theoretical analysis of $\ell_1$-regularized estimators in true high-dimensional (co)integrated datasets. Since the gram matrix based on integrated time series does not converge to a deterministic matrix, the commonly used indirect approach of imposing a restricted eigenvalue condition on the limit of the sample covariance matrix is not well-suited to high-dimensional settings. Therefore, we instead verify the condition directly on an appropriately scaled version of the gram matrix. Outside the realm of time series, a broad line of research exists that is concerned with direct verification of RE type conditions. In particular, letting $\bX$ be a $(T \times N)$-matrix with \textit{independent} rows, numerous authors have shown that $\frac{1}{T}\bX^\prime\bX$ satisfies the RE condition under Gaussian or sub-Gaussian tail conditions \citep[e.g.][]{Zhou2009,Raskutti2010,Rudelson2012,Tropp2015b,Kasiviswanathan2018}. To the best of my knowledge, however, no contributions that deal with \textit{dependent} rows exist in the current literature, let alone the case in which the rows are strongly persistent as is the case when dealing with integrated time series.

The RE condition developed in this paper enables the theoretical analysis of methods based on $\ell_1$-regularization, which we illustrate by deriving tight error bounds for the lasso when applied to cointegrated regression. We then use these error bounds to show that consistency can be maintained even when the number of potentially relevant integrated regressors grows exponentially in the sample size. Furthermore, we conjecture that the reliance on matrix concentration inequalities in the proofs offers sufficient flexibility for prospective future theoretical developments, which we motivate by providing several suggestions for valuable extensions.

The paper proceeds as follows. Section \ref{sec:REC} presents the assumptions on the data, the restricted eigenvalue and the theoretical validity of the latter. Next, in Section \ref{sec:lasso} we apply the results of the preceding section to derive error bounds and asymptotic consistency for the lasso when applied to a high-dimension (co)integrated dataset. We conclude in Section \ref{sec:conclusion}.

Finally, a word on notation. For an $N$-dimensional vector $\bx$, we define the $L_p$-norm as $\norm{\bx}_p = \left(\sum_{j=1}^N x_j^p\right)^{1/p}$. Then, for a $(T \times N)$ matrix $\bY$, the induced matrix norm is defined as $\norm{\bY}_p = \sup_{\bx \in \mathbb{R}^N \setminus \lbrace 0 \rbrace} \left\lbrace \frac{\norm{\bY\bx}_p}{\norm{\bx}_p}\right\rbrace$. If $\bY$ is square, we use $\lambda_i\left(\bY\right)$ to denote its $i$-th eigenvalue and adopt the ordering $\lambda_1\left(\bY\right) \leq \ldots \leq \lambda_N\left(\bY\right)$. Letting $A,B$ be two index sets, we define $\bx_A$ as the sub-vector of $\bx$ consisting of the elements indexed by the set $A$ and $\bY_{A,B}$ corresponds to the sub-matrix of $\bY$ that contains the rows and columns indexed by $A$ and $B$, respectively. To further simplify notation, we let $\bY_A = \bY_{A,A}$.

\section{Restricted eigenvalues for integrated time series}\label{sec:REC}

In this section, our main theoretical result is derived. In particular, we derive a finite-sample bound on the probability that the restricted eigenvalue of the gram matrix lies above a positive, universal constant. 

Let $\bS = (\bs_1,\ldots,\bs_T)^\prime$ denote a $(T\times N)$-dimensional matrix with random walks of the form $\bs_t = \sum_{s=1}^t\bepsilon_s$. The \textit{restricted minimum eigenvalue} is defined as
\begin{equation}\label{eq:restricted_eigenvalue}
\kappa(\bS,f_T,s,c_0) := \underset{\substack{J_0 \subseteq \lbrace 1,\ldots,N \rbrace, \\ \abs{J_0} \leq s}}{\text{min}} \underset{\substack{\bx \neq 0, \\ \norm{\bx_{J_0^c}}_1 \leq c_0 \norm{\bx_{J_0}}_1}}{\text{min}} \frac{f_T\norm{\bS\bx}_2}{\norm{\bx_{J_0}}_2}.
\end{equation}

We impose the following assumption on the innovations driving the random walks.
\begin{assumption}\label{assump:error_distribution}
Assume that $\bepsilon_s \overset{i.i.d.}{\sim} \mathcal{N}\left(0,\bSigma\right)$, where $0 < c_\sigma \leq \lambda_\min\left(\bSigma\right) \leq \lambda_\max\left(\bSigma\right) \leq C_\sigma < \infty$.
\end{assumption}

\begin{remark}
Assumption \ref{assump:error_distribution} imposes the innovations to be independently and identically distributed according to a multivariate Gaussian distribution. Somewhat surprisingly, the choice for a Gaussian distribution is not motivated for its exponential tail decay, although this property is exploited in the results to follow. Instead, the Gaussian distribution's invariance to orthonormal transformations is the key ingredient that enables the derivation of probabilistic bounds on restricted eigenvalues via matrix concentration inequalities.
\end{remark}

The main theoretical result brought forward in this paper concerns the behaviour of the minimum restricted eigenvalue of $\bS^\prime\bS$.
\begin{theorem}\label{Thm:REC}
Assume that  $\left\lceil\frac{c_0^2C_\sigma s}{c_\sigma} \right\rceil < N$. Then, under Assumption \ref{assump:error_distribution}, there exists $\kappa_0 > 0$ such that
\begin{equation}\label{eq:restricted_eigenvalue_probability}
    \Prob\left(\kappa\left(\bS,\frac{s\log^{3/4}N}{T},s,c_0\right) \geq \kappa_0\right) \geq 1 - (4T + C_2s)\text{e}^{-C_1s\log N},
\end{equation}
where $\kappa_0 := \kappa_0(c_0,c_\sigma,C_\sigma,C_1)$, whose expression is given in \eqref{eq:k_0}, is a constant that is independent of $s,N,T$, polynomially decreases in $c_0$, $C_\sigma$, $C_1$, and polynomially increases in $c_\sigma$.
\end{theorem}

A direct implication of Theorem \ref{Thm:REC} is that $\Prob\left(\kappa\left(\bS,\frac{s\log^{3/4}N}{T},s,c_0\right) \geq \kappa_0\right) \to 1$ whenever $T\text{e}^{-s\log N} \to 0$ and $s\text{e}^{-s\log N} \to 0$. Hence, perhaps somewhat counter-intuitively, the restricted eigenvalue condition is more likely to hold when $s$ and $N$ are large. The reason underlying this behaviour is the scaling factor of $s\log^{3/4} N$ included in the restricted minimum eigenvalue. We stress, however, that this is no free lunch, as the scaling factor will inversely affect the convergence rates of $\ell_1$-regularized estimators, as will become apparent in the following section.

\section{Lasso estimation of cointegrated data}\label{sec:lasso}

Suppose a researcher is interested in modelling a scalar-valued time series $y_t$ based on an $N$-dimensional vector time series $\bx_t$ by considering the model
\begin{equation}\label{eq:y}
\begin{split}
    \by_t = \bbeta^\prime\bx_t + \epsilon_{y,t},\\
    \bx_t = \bx_{t-1} + \epsilon_{x,t}.
\end{split}
\end{equation}
The vector $\bbeta$ is assumed to be sparse, with $S_\beta = \lbrace j : \beta_j \neq 0 \rbrace$ and $\abs{S_\beta} = s$. Hence, $\by_t$ is a unit root non-stationary time series that is cointegrated with $s$ integrated time series. We stack the innovations in the $(N+1)$ vector $\bepsilon_t = (\epsilon_{y,t},\bepsilon_{x,t}^\prime)^\prime$.

Given that $N$ may be (very) large, we consider estimating the linear model \eqref{eq:y} by the lasso. Define $\by = (y_1,\ldots,y_T)^\prime$ and $\bX = (\bx_1,\ldots,\bx_T)$. Then, the lasso estimator for $\bbeta$ is given by
\begin{equation}\label{eq:lasso}
    \hat{\bbeta} = \ \underset{\bbeta}{\text{arg min}} \norm{\by - \bX\bbeta}_2^2 + \lambda_T\norm{\bbeta}_1.
\end{equation}
Based on Theorem \ref{Thm:REC}, we are able to derive the following error bounds on the prediction and estimation errors.

%THEOREM: lASSO ERROR BOUNDS
\begin{theorem}\label{Thm:error_bounds_lasso}
Under Assumption \ref{assump:error_distribution}, it holds that
\begin{equation*}
    \norm{\bX(\hat{\bbeta}-\bbeta)}_2^2 + \lambda_T\norm{\hat{\bbeta} - \bbeta}_1 \leq \frac{4\lambda_T^2s^3\log^{3/2} N}{T^2\phi_0^2},
\end{equation*}
with probability at least
\begin{equation}\label{eq:probability_error_bound}
\begin{split}
    &1 - 2N\left[\emph{exp}\left(-\frac{2\lambda_T^2}{T^{1+2m}}\right) + 2T\emph{exp}\left(-\frac{T^{m-\frac{1}{2}}}{2C_\sigma}\right) + 2\emph{exp}\left(-\frac{\lambda_T}{4C_\sigma}\right)\right] - (4T + C_2s)e^{-C_1s\log N},
\end{split}
\end{equation}
for all $m > 0$, and universal constants $C_1,C_2>0$.
\end{theorem}

The probabilistic finite-sample error bound in Theorem \ref{Thm:error_bounds_lasso} elucidates the dual role of the penalty parameter $\lambda_T$. Whilst a high degree of penalization results in a looser error bound, it simultaneously increases the probability with which this  error bound holds. The attentive reader may note that the probability in \eqref{eq:probability_error_bound} contains the lower bound on the probability that the REC holds from Theorem \ref{Thm:REC}, minus an additional term. The latter term resembles the probability that $\lambda_T$ fails to dominate the empirical process, and is the decisive factor for determining the minimum amount of penalization required to ensure consistency. Indeed, based on Theorem \ref{Thm:error_bounds_lasso}, the consistency of the lasso estimator can be established under a suitable choice of $\lambda_T$ and appropriate restrictions on the growth rates of $T,N$, and $s$.

\begin{corollary}\label{Cor:lasso_consistency}
Assume that,as $T,N,s \to \infty$, $\frac{T^{1+\xi}\log^{1/2} N}{\lambda_T} \to 0$  and $\frac{\log N}{T^\xi} \to 0$ for some $\xi>0$, and $\frac{\log T}{s \log N} \to 0$. Then, under Assumption \ref{assump:error_distribution}, it holds that
\begin{equation*}
    \norm{\hat{\bbeta} - \bbeta}_1 = O_p\left(\frac{s^3\log^2 N}{T^{1-\xi}}\right),
\end{equation*}
for any $\xi>0$.
\end{corollary}

Corollary \ref{Cor:lasso_consistency} makes the rate of consistency, as well as the sufficient restrictions on growth rates of $\lambda_T,T,N$ and $s$ explicit. Most noteworthy is that consistency may be attained when $N$ grows exponentially in $T$, thereby providing a first justification for the use of the lasso in ultra high-dimensional cointegrated applications. Furthermore, we note that when the number of time series that cointegrate with the dependent variable ($s$) grows at a slow polynomial rate in $T$, while the total number of candidate time series $N$ grows at any polynomial rate in $T$, then $\hat{\bbeta}$ almost attains attains the same $T$-consistency as in the fixed-dimensional setting.

\section{Conclusion}\label{sec:conclusion}

In this paper, we validate the restricted eigenvalue condition (REC) proposed by \citet{Bickel2009} in a stationary setting to unit root nonstationary data. The gram matrix is scaled by $\frac{s^2\log^{3/2}N}{T^2}$ instead of the usual $\frac{1}{T^2}$ that is familiar from the fixed-dimensional case. We apply our extended REC to derive finite-sample error bounds on the lasso when applied to high-dimensional cointegrating regressions. The bounds reveal that consistency can be attained, even under exponential growth of the number of time series. In addition, when the number of potentially relevant integrated regressors grow at any polynomial rate in $T$, while the number of regressors that cointegrate with the dependent variable grows at a slow polynomial rate, the estimator is almost $T$-consistent, as is the case in the classic fixed-dimensional setting.

A natural question that may arise concerns the validity of these results in the presence of multiple cointegrating relationships. We are optimistic that the results in this paper can serve useful when deriving the properties of system estimators such as penalized vector error correction models, and we consider this an interesting avenue for future research.

\begin{appendices}

\section{Proofs}
\begin{small}

\subsection{Lemmas}

\begin{lemma}[Theorem 1.1 of \citet{Tropp2012}]\label{Lemma:Chernoff}
Consider a finite sequence $\lbrace\bX_t\rbrace$ of independent, random, self-adjoint matrices with dimension $N$. Assume that each random matrix satisfies
\begin{equation*}
\bX_t \succeq 0, \ \text{ and } \ \lambda_\max\left(\bX_t\right) \leq R  \ \text{ almost surely}.
\end{equation*}
Define
\begin{equation*}
\mu_\min:= \lambda_\min\left(\sum_t \E(\bX_t)\right) \quad \text{and} \quad \mu_\max := \lambda_\max\left(\sum_t \E(\bX_t)\right).
\end{equation*}
Then,
\begin{align}
&\Prob\left(\lambda_\min\left(\sum_{t=1}^T \bX_t\right) \leq (1-\delta)\mu_\min\right) \leq N\left(\frac{e^{-\delta}}{(1-\delta)^{1-\delta}}\right)^{\mu_\min/R} \text{ for } \delta \in [0,1]\text{, and} \label{eq:Chernoff_min}\\
&\Prob\left(\lambda_\max\left(\sum_{t=1}^T \bX_t\right) \geq (1+\delta)\mu_\max\right) \leq  N\left(\frac{e^{\delta}}{(1+\delta)^{1+\delta}}\right)^{\mu_\min/R} \text{ for } \delta \geq 0.\label{eq:Chernoff_max}
\end{align}
\end{lemma}

\begin{lemma}\label{Lemma:emp_proc}
Let $f_T = \frac{s\log^{3/4} N}{T}$ and define the set $\mathcal{A}_T(a) = \left\lbrace f_T^2\norm{\bX^\prime\bepsilon_y}_\infty \geq a \right\rbrace$. Then, under Assumption \ref{assump:error_distribution}, it holds that
\begin{equation}
    \Prob\left(\mathcal{A}_T(a)\right) \geq 1 - 2N\left[\emph{exp}\left(-\frac{2a^2T^{3-2m}}{s^4\log^4 N}\right) + 2T\emph{exp}\left(-\frac{T^{m-\frac{1}{2}}}{2C_\sigma}\right) + 2\emph{exp}\left(-\frac{aT^2}{4C_\sigma s^2\log^2 N}\right)\right],
\end{equation}
for any $m \in \mathbb{R}$.
\end{lemma}
\begin{proof}
The main argument to prove Lemma \ref{Lemma:emp_proc} relies on a truncation argument in combination with the Azuma-Hoeffding bound and is inspired by the proof of Theorem 1 in \citet{KockCallot2015}. First, by application of the union bound,
\begin{equation}\label{eq:prob_bound1}
\begin{split}
    &\Prob\left(\max_{1 \leq i \leq N}\abs{\sum_{t=1}^T x_{i,t}\epsilon_{y,t}} \geq \frac{a}{f_T^2}\right)\\
    &\leq \Prob\left(\max_{1 \leq i \leq N}\abs{\sum_{t=1}^T x_{i,t-1}\epsilon_{y,t}} \geq \frac{a}{2f_T^2}\right) + \Prob\left(\max_{1 \leq i \leq N}\abs{\epsilon_{i,t}\epsilon_{y,t}} \geq \frac{a}{2f_T^2}\right)
\end{split}
\end{equation}
We bound each of the two RHS terms in \eqref{eq:prob_bound1} individually. For the first term, we obtain
\begin{equation}\label{eq:prob_bound_term1}
\begin{split}
    &\Prob\left(\max_{1 \leq i \leq N}\abs{\sum_{t=1}^T x_{i,t-1}\epsilon_{y,t}} \geq \frac{a}{2f_T^2}\right)\\
    &\leq \Prob\left(\max_{1 \leq i \leq N}\abs{\sum_{t=1}^T x_{i,t-1}\epsilon_{y,t}} \geq \frac{a}{2f_T^2}, \max_{1 \leq i \leq N}\max_{1 \leq t \leq T}\abs{x_{i,t-1}\epsilon_{y,t}} \leq T^m\right)\\
    &\quad + \Prob\left(\max_{1 \leq i \leq N}\max_{1 \leq t \leq T} \abs{x_{i,t-1}\epsilon_{y,t}} > T^m\right)\\
    &\leq \sum_{i=1}^N\left[\Prob\left(\abs{\sum_{t=1}^T x_{i,t-1}\epsilon_{y,t}} \geq \frac{a}{2f_T^2},\max_{1 \leq t \leq T}\abs{x_{i,t-1}\epsilon_{y,t}} \leq T^m\right) + \Prob\left(\max_{1 \leq t \leq T} \abs{x_{i,t-1}\epsilon_{y,t}} > T^m\right)\right]
\end{split}
\end{equation}
To bound the first term on the RHS of \eqref{eq:prob_bound_term1}, we aim to apply the Azuma-Hoeffding inequality for bounded martingale difference sequences \citep[e.g.][Corollary 2.20]{Wainwright2019} . To this end, let $\mathcal{F}_t$ denote the natural filtration with respect to the stochastic process $\lbrace\bepsilon_j\rbrace_{j=-\infty}^t$. By the proof of Lemma 4 in \citet{KockCallot2015}, it follows that
\begin{equation*}
\E\left(x_{i,t-1}\epsilon_{y,t}\mathbbm{1}\lbrace \abs{x_{i,t-1}\epsilon_{y,t}} \leq T^m\rbrace | \mathcal{F}_{t-1}\right) = 0.
\end{equation*}
Then, by the Azuma-Hoeffding inequality,
\begin{equation}\label{eq:azuma_hoeffding_bound}
\begin{split}
    &\sum_{i=1}^N\Prob\left(\abs{\sum_{t=1}^T x_{i,t-1}\epsilon_{y,t}} \geq \frac{a}{2f_T^2},\max_{1 \leq t \leq T}\abs{x_{i,t-1}\epsilon_{y,t}} \leq T^m\right)\\
    &\leq 2N\text{exp}\left(\frac{-2a^2}{T^{2m+1}f_T^4}\right) = 2N\text{exp}\left(\frac{-2a^2T^{3-2m}}{s^4\log^3 N}\right).
\end{split}
\end{equation}
To bound the second term in \eqref{eq:prob_bound_term1}, we first note that for any two random variables $x,y$ and constant $c>0$, it holds that
\begin{equation*}
    \Prob(\abs{xy} > c) \leq \Prob(\abs{x}c^k > c,\abs{y} \leq c^k) + \Prob(\abs{y} > c^k) \leq \Prob(\abs{x} > c^{1-k}) + \Prob(\abs{y} > c^k).
\end{equation*}
In addition, note that $\bx_{i,t-1} \sim \mathcal{N}(0,t\sigma_{ii}^2)$ and $\epsilon_{y,t} \sim \mathcal{N}(0,\sigma_{11}^2)$ by Assumption \ref{assump:error_distribution}. Then, based on a concentration inequality for sub-Gaussian variables \citep[e.g.][equation (2.9)]{Wainwright2019}, we bound
\begin{equation}\label{eq:sub-gaussian_tail}
\begin{split}
    \Prob\left(\abs{x_{i,t-1}\epsilon_{y,t}} > T^m\right) &\leq \Prob\left(\abs{x_{i,t-1}} > T^{mk}\right) + \Prob\left(\abs{\epsilon_{y,t}} > T^{m(1-k)}\right)\\
    &\leq 2\text{exp}\left(-\frac{T^{2mk-1}}{2\sigma_{ii}^2}\right) + 2\text{exp}\left(-\frac{T^{2m(1-k)}}{2\sigma_{11}^2}\right).
\end{split}
\end{equation}
We optimize the rate at which \eqref{eq:sub-gaussian_tail} converges to zero by equating the exponents:
\begin{equation*}
    2mk-1 = 2m(1-k) \Rightarrow k = \frac{1}{2} + \frac{1}{4m}.
\end{equation*}
Plugging in this value for $k$, together with the upper bound $(\sigma_{ii}^2, \vee \sigma_{11}^2) \leq C_\sigma$ for all $1 \leq i \leq N+1$ and a union bound, it follows that
\begin{equation}\label{eq:max_xepsilon_bound}
    \sum_{i=1}^N\Prob\left(\max_{1 \leq t \leq T}\abs{x_{i,t-1}\epsilon_{y,t}} > T^m\right) \leq 4NT\text{exp}\left(-\frac{T^{m-\frac{1}{2}}}{2C_\sigma}\right).
\end{equation}
Next, again using the sub-Gaussian concentration inequality, we bound the second term of \eqref{eq:prob_bound1} as
\begin{equation}\label{eq:max_epsilons_bound}
\begin{split}
     \Prob\left(\max_{1 \leq i \leq N}\abs{\epsilon_{i,t}\epsilon_{y,t}} \geq \frac{a}{2f_T^2}\right) &\leq \sum_{i=1}^N \left[\Prob\left(\abs{\epsilon_{i,t}} \geq \left(\frac{a}{2f_T^2}\right)^{1/2}\right) + \Prob\left(\abs{\epsilon_{y,t}} \geq \left(\frac{a}{2f_T^2}\right)^{1/2}\right)\right]\\
     &\leq 4N\text{exp}\left(-\frac{a}{4f_T^2C_\sigma}\right) = 4N\text{exp}\left(-\frac{aT^2}{4C_\sigma s^2\log^{3/2} N}\right).
\end{split}
\end{equation}
Finally, combining \eqref{eq:azuma_hoeffding_bound}, \eqref{eq:max_xepsilon_bound} and \eqref{eq:max_epsilons_bound}, we bound the RHS of \eqref{eq:prob_bound1} as
\begin{equation*}
\begin{split}
    &\Prob\left(\max_{1 \leq i \leq N}\abs{\sum_{t=1}^T x_{i,t}\epsilon_{y,t}} \geq \frac{a}{f_T^2}\right)\\
    &\leq  2N\left[\text{exp}\left(-\frac{2a^2T^{3-2m}}{s^4\log^3 N}\right) + 2T\text{exp}\left(-\frac{T^{m-\frac{1}{2}}}{2C_\sigma}\right) + 2\text{exp}\left(-\frac{aT^2}{4C_\sigma s^2\log^{3/2} N}\right)\right].
\end{split}
\end{equation*}
\end{proof}

\subsection{Proofs of main results}

\begin{proof}[\textbf{Proof of Theorem \ref{Thm:REC}}]
We first show that the restricted minimum eigenvalue condition holds for a transformed matrix. Rewrite $\bS = \bU\bE$, where $\bU$ is a lower triangular matrix with ones on and below the diagonal and $\bE = (\bepsilon_1,\ldots,\bepsilon_T)^\prime$. Decompose $\bU^\prime\bU = \bV\bLambda\bV^\prime$, where $\bLambda = \diag(\lambda_1,\ldots,\lambda_T)$ are the eigenvalues of $\bU^\prime\bU$. By Lemma 1 in \citet{Akesson1998}, it holds that
\begin{equation}\label{eq:eig_UU1}
\lambda_t^{-1} = 4\sin^2\left(\frac{\omega_t}{2}\right) = 2(1-\cos \omega_t),
\end{equation}
with $\omega_t = \frac{(2t - 1)\pi}{2T+1}$.\footnote{The second equality in \eqref{eq:eig_UU1} is based on the identity $\cos(2\alpha) = 1 - 2\sin^2(\alpha)$.} Both $\lambda_1 = O(T^2)$ and $\sum_{t=1}^T\lambda_t = O(T^2)$, which turns out to be troublesome when deriving the behaviour of the spectrum. Accordingly, we consider a transformed matrix
\begin{equation*}
\bS^{(\phi)^\prime}\bS^{(\phi)} = \sum_{t=1}^T\lambda_{\phi,t}\tilde{\bepsilon}_t\tilde{\bepsilon}t^\prime,
\end{equation*}
with
\begin{equation*}
\lambda^{-1}_{\phi,t} = 2(1 + \phi - \cos \omega_t),
\end{equation*}
where $\phi$ is a positive sequence that decreases to zero as $T,N,s \to \infty$. Note that
\begin{equation}\label{eq:SS-SS_phi}
\bS^\prime\bS - \bS^{(\phi)^\prime}\bS^{(\phi)} = \sum_{t=1}^T (\lambda_t-\lambda_{\phi,t})\tilde{\bepsilon}_t\tilde{\bepsilon}_t^\prime \succeq 0,
\end{equation}
thus implying that $\kappa\left(\bS,s,c_0\right) > \kappa\left(\bS^{(\phi)},s,c_0\right)$. Define the sparse minimum and maximum eigenvalues as
\begin{equation}\label{eq:restricted_eigs1}
\begin{split}
\phi_\min(\bS^{(\phi_u)},f_T,u) &= \underset{\bx \in \mathbb{R}^N: 1 \leq \mathcal{M}(\bx) \leq u}{\text{min}} \frac{\norm{f_T\bS^{(\phi_u)}\bx}_2^2}{\norm{\bx}_2^2},\\
\phi_\max(\bS^{(\phi_u)},f_T,u) &= \underset{\bx \in \mathbb{R}^N: 1 \leq \mathcal{M}(\bx) \leq u}{\text{max}} \frac{\norm{f_T\bS^{(\phi_u)}\bx}_2^2}{\norm{\bx}_2^2}.
\end{split}
\end{equation}
By Lemma 4.1(ii) in \citet{Bickel2009}, it follows that
\begin{equation}\label{eq:Bickel_ineq}
    \kappa\left(\bS^{(\phi_s)},f_T,s,c_0\right) \geq 
    \sqrt{\phi_\min(\bS^{(\phi_s)},f_T,s+m)} - c_0\sqrt{\phi_\max(\bS^{(\phi_s)},f_T,m)}\sqrt{\frac{s}{m}},
\end{equation}
where $m \geq s$ and $s + m \leq N$. Accordingly, we aim the provide a lower bound for \eqref{eq:Bickel_ineq} by deriving probabilistic upper and lower bounds for the sparse minimum and maximum eigenvalues, respectively.

Let the collection of all index sets of cardinality $s$ be denoted by
\begin{equation*}
\mathcal{A}(s) := \left\lbrace A \subseteq \lbrace 1, \ldots, N\rbrace :  \abs{A} = s\right\rbrace,
\end{equation*}
and note that its cardinality is upper bounded by
\begin{equation}\label{eq:A_card}
\abs{\mathcal{A}(s)} \leq \left(\frac{eN}{s}\right)^s = e^{s\log(eN/s)} \leq e^{s(1 - \log s +\log N)} \leq e^{s\log N},
\end{equation}
for $s \geq 3$. For any $A \in \mathcal{A}(s)$, let $\bS^{(\phi)}_A = \bU^{(\phi)}\bE_A$ represent the columns of $\bS^{(\phi)}$ indexed by $A$. Then, an equivalent definition of the sparse eigenvalues in \eqref{eq:restricted_eigs1} is given by
\begin{equation}\label{eq:restricted_eigs2}
\begin{split}
&\phi_\min(\bS^{(\phi_s)},f_T,s) = \underset{A \in \mathcal{A}(s)}{\text{min}}\lambda_\min\left(f_T^2\bS_A^{(\phi_s)\prime}\bS_A^{(\phi_s)}\right),\\
&\phi_\max(\bS^{(\phi_s)},f_T,s) = \underset{A \in \mathcal{A}(s)}{\text{max}}\lambda_\max\left(f_T^2\bS_A^{(\phi_s)\prime}\bS_A^{(\phi_s)}\right).
\end{split}
\end{equation}

We proceed to derive the stochastic order of the quantities in \eqref{eq:restricted_eigs2}. For any $A \in \mathcal{A}(s)$ it follows that
\begin{equation*}
\bS_A^{(\phi_s)\prime}\bS_A^{(\phi_s)} = \bE_A^\prime\bV\bLambda^{(\phi_s)}\bV^\prime\bE_A = \tilde{\bE}_A^\prime\bLambda^{(\phi_s)}\tilde{\bE}_A = \sum_{t=1}^T\lambda_{\phi_s,t}\tilde{\bepsilon}_{A,t}\tilde{\bepsilon}_{A,t}^\prime,
\end{equation*}
where $\tilde{\bepsilon}_{A,t} \overset{i.i.d.}{\sim} \mathcal{N}(\bm{0},\bSigma_{A})$ by the rotational invariance of the multivariate normal distribution. Let
\begin{equation*}
\begin{split}
    \mu^{\phi_s}_\min(A) &:= \lambda_\min\left(\sum_{t=1}^T \lambda_{\phi_s,t}\E\left(\tilde{\bepsilon}_{A,t}\tilde{\bepsilon}_{A,t}^\prime\right)\right) = \lambda_\min\left(\bSigma_A\right)\sum_{t=1}^T\lambda_{\phi_s,t}\\
    \mu^{\phi_s}_\max(A) &:= \lambda_\max\left(\sum_{t=1}^T \lambda_{\phi_s,t}\E\left(\tilde{\bepsilon}_{A,t}\tilde{\bepsilon}_{A,t}^\prime\right)\right) = \lambda_\max\left(\bSigma_A\right)\sum_{t=1}^T\lambda_{\phi_s,t},
\end{split}
\end{equation*}
and the corresponding quantities
\begin{equation*}
    \begin{split}
        \mu^s_\min &= \min_{A \in \mathcal{A}(s)} \mu^{\phi_s}_\min (A),\\
        \mu^s_\max &= \max_{A \in \mathcal{A}(s)} \mu^{\phi_s}_\max (A).
    \end{split}
\end{equation*}
In addition, define 
\begin{equation}\label{eq:R}
R_s = C_\sigma^2\lambda_{\phi_s,1}s\left((C_1 + 1)\log^{1/2} N+1\right),
\end{equation}
for an arbitrary constant $C_1 > 0$. First, we use the truncation argument
\begin{equation}\label{eq:min_eig_bound_1}
\begin{split}
&\Prob\left(\min_{A \in \mathcal{A}(s)} \lambda_\min\left(\frac{1}{\mu_\min(A)}\bS_A^{(\phi_s)^\prime}\bS_A^{(\phi_s)}\right) \leq 1-\delta\right)\\
& \leq \Prob\left(\min_{A \in \mathcal{A}(s)} \lambda_\min\left(\frac{1}{\mu_\min(A)}\bS_A^{(\phi_s)^\prime}\bS_A^{(\phi_s)}\right) \leq 1-\delta,\sup_{A \in \mathcal{A}(s)} \sup_{1\leq t\leq T}\lambda_\max\left(\lambda_{\phi_s,t}\tilde{\bepsilon}_{A,t}\tilde{\bepsilon}_{A,t}^\prime\right)\leq R_s \right)\\
&\quad + \Prob\left(\sup_{A \in \mathcal{A}(s)} \sup_{1\leq t\leq T}\lambda_\max\left(\lambda_{\phi_s,t}\tilde{\bepsilon}_{A,t}\tilde{\bepsilon}_{A,t}^\prime\right) \geq R_s \right).\\
\end{split}
\end{equation}

The first term on the RHS can now be bounded with the use of the Chernoff bound in Lemma \ref{Lemma:Chernoff}. Define the constant
\begin{equation*}
K_\delta = \max\left(\frac{e^{-\delta}}{(1-\delta)^{1-\delta}}, \frac{e^{\delta}}{(1+\delta)^{1+\delta}}\right)
\end{equation*}
and note that $0 < K_\delta < 1$ for any $0 < \delta < 1$. Then, Lemma \ref{Lemma:Chernoff} states that, for each $A \in \mathcal{A}(s)$,
\begin{equation}\label{eq:Chernoff2}
\begin{split}
\Prob\left(\lambda_\min\left(\frac{1}{\mu_\min(A)}\bS_A^{(\phi_s)^\prime}\bS_A^{(\phi_s)}\right) \leq 1-\delta, \sup_{1\leq t\leq T}\lambda_\max\left(\lambda_{\phi_s,t}\tilde{\bepsilon}_{A,t}\tilde{\bepsilon}_{A,t}^\prime\right)\leq R_s \right) \leq s K_\delta^{\mu^{\phi_s}_\min(A)/R_s}.
\end{split}
\end{equation}
For any $s \leq N$, we can bound $\mu^s_\min$ from below as
\begin{equation}\label{eq:eig_mu}
\begin{split}
2\mu^s_\min &\geq \sum_{t=1}^T \frac{c_\sigma}{1-\cos \omega_t + \phi_s} \geq \sum_{t=1}^T \frac{c_\sigma}{\left(\frac{(2t+1)\pi}{2T+1}\right)^2 + \phi_s}\\
&= \sum_{t=1}^T \frac{c_\sigma(2T+1)^2}{(2t+1)^2\pi^2 + \phi_s(2T+1)^2}.
\end{split}
\end{equation}
Since $\phi_s$ is defined as a positive sequence decreasing to zero in $T$, we have that $\phi_s < 1$ for large enough $T$. It follows that
\begin{equation}\label{eq:eig_mu2}
\begin{split}
\mu^s_\min &\geq \sum_{t=1}^T \frac{c_\sigma(2T+1)^2}{2\left((2t+1)^2\pi^2 + \phi_s(2T+1)^2\right)} \geq \sum_{t=1}^T \frac{c_\sigma(2T)^2}{2\left((3t)^2\pi^2 + \phi_s(3T)^2\pi^2\right)}\\
&=\sum_{t=1}^T \frac{4c_\sigma T^2}{18\pi^2\left(t^2 + \phi_sT^2\right)} \geq \sum_{t=1}^{\left[T\phi_s^{1/2}\right]} \frac{4c_\sigma T^2}{18\pi^2\left(t^2 + \phi_sT^2\right)}\\
&\geq \frac{4c_\sigma \phi_s^{1/2}T^3}{36\pi^2\phi_sT^2} = \frac{c_\sigma T}{9\pi^2\phi_s^{1/2}}.
\end{split}
\end{equation}
Next, bound $\mu^s_\min/R_s$ as
\begin{equation}\label{eq:mu/r}
\begin{split}
\frac{\mu^s_\min}{R_s} &\geq \frac{c_\sigma T}{9\pi^2C_\sigma^2\phi_s^{1/2}\lambda_{\phi,1}s\left((C_1 + 1)\log^{1/2} N+1\right)} = \frac{c_\sigma T(1-\cos \omega_1 + \phi_s)}{9\pi^2C_\sigma^2\phi_s^{1/2}s\left((C_1 + 1)\log^{1/2} N+1\right)}\\
&\overset{(*)}{\geq} \frac{c_\sigma T\phi_s^{1/2}}{9\pi^2C_\sigma^2(C_1 + 3)s\log^{1/2} N} \overset{(**)}{\geq} -\frac{(C_1+1)s\log N}{\log K_\delta},
\end{split}
\end{equation}
where (*) follows from the observation that $2\log^{1/2} N > 1$ for any $N \geq 2$ and (**) follows after setting
\begin{equation}\label{eq:phi}
  \phi_s^{1/2} = -\frac{9\pi^2(C_1+1)(C_1+3)C_\sigma^2 s^2\log^{3/2} N}{c_\sigma T\log K_\delta}. 
\end{equation}
Note that $\phi_s^{1/2}$ is positive due to the negative term $\log K_\delta$ in the denominator. Then, combining \eqref{eq:eig_mu2}-\eqref{eq:phi}, we use the Chernoff matrix-bound from Lemma \ref{Lemma:Chernoff} to bound the first RHS term of \eqref{eq:min_eig_bound_1} as
\begin{equation}\label{eq:Chernoff_upper_bound}
\begin{split}
    &\Prob\left(\min_{A \in \mathcal{A}(s)} \lambda_\min\left(\frac{1}{\mu^{\phi_s}_\min(A)}\bS_A^{(\phi)^\prime}\bS_A^{(\phi)}\right) \leq 1-\delta,\sup_{A \in \mathcal{A}(s)} \sup_{1\leq t\leq T}\lambda_\max\left(\lambda_{\phi_s,t}\tilde{\bepsilon}_t\tilde{\bepsilon}_t^\prime\right)\leq R_s \right)\\
    &\leq \sum_{A \in \mathcal{A}(s)}\Prob\left(\min_{A \in \mathcal{A}(s)} \lambda_\min\left(\frac{1}{\mu^{\phi_s}_\min(A)}\bS_A^{(\phi)^\prime}\bS_A^{(\phi)}\right) \leq 1-\delta,\sup_{A \in \mathcal{A}(s)} \sup_{1\leq t\leq T}\lambda_\max\left(\lambda_{\phi_s,t}\tilde{\bepsilon}_t\tilde{\bepsilon}_t^\prime\right)\leq R_s \right)\\
    &\leq s\cdot\text{exp}\left(s \log N + \frac{\mu^s_\min \log K_\delta}{R_s}\right) \leq s\cdot\text{exp}\left(-C_1 s \log N\right).
\end{split}
\end{equation}

Next, we proceed by bounding the second RHS term of \eqref{eq:min_eig_bound_1}. Let $\bv_t \sim  \mathcal{N}\left(0,\bI_{\abs{A}}\right)$, such that $\tilde{\bepsilon}_{A,t} \overset{d}{=} \bSigma_A^{1/2}\bv_t$. By application of the union bound, the fact that the $\tilde{\bepsilon}_t$ are identically distributed, the definition of $R_s$ and the triangle inequality, it holds that
\begin{equation}\label{eq:r}
\begin{split}
&\Prob\left(\sup_{A \in \mathcal{A}(s)}\sup_{1\leq t \leq T} \lambda_\max\left(\lambda_{\phi_s,t}\tilde{\bepsilon}_{A,t}\tilde{\bepsilon}_{A,t}^\prime\right) \geq R_s\right) \leq \sum_{A \in \mathcal{A}(s)}\sum_{t=1}^T \Prob\left(\lambda_\max\left(\tilde{\bepsilon}_{A,t}\tilde{\bepsilon}_{A,t}^\prime\right) \geq \frac{R_s}{\lambda_{\phi_s,t}}\right)\\
&\leq T\sum_{A \in \mathcal{A}(s)}\Prob\left(\lambda_\max\left(\tilde{\bepsilon}_{A,1}\tilde{\bepsilon}_{A,1}^\prime\right) \geq \frac{R_s}{\lambda_{\phi_s,1}}\right) = T\sum_{A \in \mathcal{A}(s)}\Prob\left(\norm{\tilde{\bepsilon}_{A,1}}_2^2 \geq \frac{R_s}{\lambda_{\phi_s,1}}\right)\\
&\leq T\sum_{A \in \mathcal{A}(s)}\Prob\left(\norm{\tilde{\bv}_1}_2^2 \geq \frac{R_s}{\lambda_{\phi_s,1}\norm{\bSigma_A}_2^2}\right) \leq T\sum_{A \in \mathcal{A}(s)}\Prob\left(\frac{1}{s}\sum_{j \in A}v^2_{j,1} \geq (C_1 + 1)\log^{1/2} N+1\right)\\
& \leq T\sum_{A \in \mathcal{A}(s)}\Prob\left(\abs{\frac{1}{s}\sum_{j \in A}v^2_{j,1} - 1} \geq (C_1 + 1)\log^{1/2} N\right) \overset{(*)}{\leq} 2T\exp\left(-C_1s\log N\right),
\end{split}
\end{equation}
where (*) follows by the sub-exponential concentration of property of  $\chi^2$ random variables (see Example 2.11 in \citet{Wainwright2019}). Finally, plugging \eqref{eq:Chernoff_upper_bound} and \eqref{eq:r} into \eqref{eq:min_eig_bound_1}, we obtain the bound for the minimum sparse eigenvalue as
\begin{equation}\label{eq:min_eig_bound}
\begin{split}
&\Prob\left(\phi_\min\left(\bS^{(\phi_s)},\frac{1}{\sqrt{\mu^s_\min}},s\right) \leq 1-\delta\right) \leq \Prob\left(\min_{A \in \mathcal{A}(s)} \lambda_\min\left(\frac{1}{\mu^{\phi_s}_\min(A)}\bS_A^{(\phi_s)^\prime}\bS_A^{(\phi_s)}\right) \leq 1-\delta\right)\\
&\leq (2T+s)\text{exp}\left(-C_1 s \log N\right).
\end{split}
\end{equation}
Finally, since the lower bounds in \eqref{eq:eig_mu2} and \eqref{eq:mu/r} remain valid for $\mu_\max$ and $\frac{\mu_\max}{R_s}$, respectively, it follows by analogous reasoning that the maximum sparse eigenvalue is bounded as
\begin{equation}\label{eq:max_eig_bound}
\begin{split}
&\Prob\left(\phi^2_\max\left(\bS^{(\phi_s)},\frac{1}{\sqrt{\mu^s_\max}},s\right) \geq 1+\delta\right) \leq \Prob\left(\max_{A \in \mathcal{A}(s)} \lambda_\max\left(\frac{1}{\mu^{\phi_s}_\max(A)}\bS_A^{(\phi_s)^\prime}\bS_A^{(\phi_s)}\right) \geq 1+\delta\right)\\
&\leq  (2T+s)\text{exp}\left(-C_1 s \log N\right).
\end{split}
\end{equation}

Having derived probabilistic bounds for the minimum and maximum sparse eigenvalues, we proceed to prove the restricted eigenvalue condition in Theorem \ref{Thm:REC}. Define the set
\begin{equation*}
    \mathcal{R}(s,m) = \left\lbrace \phi_\min\left(\bS^{(\phi_{s+m})},\frac{1}{\sqrt{\mu^{s+m}_\min}},s+m\right) > 1-\delta\right\rbrace \bigcap \left\lbrace\phi_\max\left(\bS^{(\phi_s)},\frac{1}{\sqrt{\mu^s_\max}},s\right) < 1+\delta\right\rbrace.
\end{equation*}
For a positive constant $\kappa > 0$, define 
\begin{equation*}
    m_\kappa = \left\lceil\frac{c_0^2C_\sigma(1+\delta)s}{c_\sigma(\sqrt{1-\delta} - \kappa)^2}\right\rceil =: \left\lceil C_\kappa s\right\rceil,
\end{equation*}
with $C_\kappa = \frac{c_0^2C_\sigma(1+\delta)}{c_\sigma(\sqrt{1-\delta} - \kappa)^2}$. Then, continuing from \eqref{eq:Bickel_ineq}, it holds on $\mathcal{R}\left(s,m_\kappa\right)$ that
\begin{equation}
\begin{split}
        &\kappa\left(\bS^{(\phi_{s})},\frac{1}{\sqrt{\mu^{s+m_\kappa}_\min}},s,c_0\right)\\
        &\geq \sqrt{\phi_\min\left(\bS^{(\phi_s)},\frac{1}{\sqrt{\mu^{s+m_\kappa}_\min}},s+m_\kappa\right)} - c_0\sqrt{\phi_\max\left(\bS^{(\phi_s)},\frac{1}{\sqrt{\mu^{s+m_\kappa}_\min}},s\right)}\sqrt{\frac{s}{m_\kappa}}\\
        &\overset{(*)}{\geq} \sqrt{\phi_\min\left(\bS^{(\phi_{s + m_\kappa})},\frac{1}{\sqrt{\mu^{s+m_\kappa}_\min}},s+m_\kappa\right)} - c_0\sqrt{\phi_\max\left(\bS^{(\phi_s)},\frac{1}{\sqrt{\mu^{s+m_\kappa}_\min}},s\right)}\sqrt{\frac{s}{m_\kappa}}\\
        &\geq \sqrt{\phi_\min\left(\bS^{(\phi_{s + m_\kappa})},\frac{1}{\sqrt{\mu^{s+m_\kappa}_\min}},s+m_\kappa\right)} - c_0\sqrt{\frac{\mu^s_\max}{\mu^{s+m_\kappa}_\min}}\sqrt{\phi_\max\left(\bS^{(\phi_s)},\frac{1}{\sqrt{\mu^s_\max}},s\right)}\sqrt{\frac{s}{m_\kappa}}\\
        &\overset{(**)}{\geq} \sqrt{1-\delta} - c_0\sqrt{\frac{C_\sigma(1+\delta)s}{c_\sigma m_\kappa}} \geq \kappa,
\end{split}
\end{equation}
where ($*$) uses that $\lambda_{\phi_s,t} > \lambda_{\phi_{s+m},t}$ and (**) uses that
\begin{equation*}
    \frac{\mu^s_\max}{\mu^{s+m_\kappa}_\min} \leq \frac{\mu^{s+m_\kappa}_\max}{\mu^{s+m_\kappa}_\min} \leq \frac{C_\sigma}{c_\sigma}.
\end{equation*} 
Note that by \eqref{eq:eig_mu2} and \eqref{eq:phi}
\begin{equation*}
    \frac{1}{\mu^{s+m_\kappa}_\min} \leq \frac{9\pi^2\phi_{s+m_\kappa}^{1/2}}{c_\sigma T} = -\frac{81\pi^4(C_1 + 1)(C_1 + 3)C_\sigma^2 (C_\kappa+1)^2s^2 \log^{3/2} N}{c_\sigma^2 T^2 \log K_\delta} = \frac{C_\mu s^2 \log^{3/2} N}{T^2},
\end{equation*}
with $C_\mu = -\frac{81\pi^4(C_1 + 1)(C_1 + 3)C_\sigma^2 (C_\kappa + 1)^2}{c_\sigma^2\log K_\delta}$. Then, we conclude that
\begin{equation}\label{eq:REC_phi}
    \kappa\left(\bS,\frac{s\log^{3/4}N}{T},s,c_0\right) \geq C_\mu^{-1/2}\kappa\left(\bS^{(\phi_s)},\frac{1}{\sqrt{\mu^{s+m_\kappa}_\min}},s,c_0\right) \geq C_\mu^{-1/2}\kappa =: \kappa_0
\end{equation}
on the set $\mathcal{R}\left(s,m_\kappa\right)$, where the fully expanded expression for $\kappa_0$ reads
\begin{equation}\label{eq:k_0}
    \kappa_0 := \kappa_0(c_0,c_\sigma,C_\sigma,C_1) = -\frac{\kappa c_\sigma\log^{1/2} K_\delta}{9\pi^2\sqrt{(C_1 + 1)(C_1 + 3)}C_\sigma \left(\frac{c_0^2C_\sigma(1+\delta)}{c_\sigma(\sqrt{1-\delta} - \kappa)^2} + 1\right)},
\end{equation}
where $\kappa > 0$ and $0 < \delta < 1$ are free constants chosen such that $\left\lceil\frac{c_0^2C_\sigma(1+\delta)s}{c_\sigma(\sqrt{1-\delta} - \kappa)^2}\right\rceil < N$.

Finally, based on \eqref{eq:min_eig_bound} and \eqref{eq:max_eig_bound},
\begin{equation*}
\begin{split}
    \Prob\left(\mathcal{R}(s,m_\kappa)\right) &\geq 1 -  (2T + s)\text{e}^{-C_1s \log N}  - (2T + m_\kappa)\text{e}^{-C_1m_\kappa \log N}\\
    &\geq 1 - (4T + (C_k + 1)s)\text{e}^{-C_1s\log N} = 1 - (4T + C_2s)\text{e}^{-C_1s\log N},
\end{split}
\end{equation*}
with $C_2 = C_\kappa + 1$. This completes the proof of Theorem \ref{Thm:REC}.

\end{proof}
\begin{proof}[\textbf{Proof of Theorem \ref{Thm:error_bounds_lasso}}]
The proof of the error bound is based on standard arguments, see for example Theorem 6.1 in \citet{Buhlmann2011}, but the probability on which it holds requires new arguments. First, let $f_T = \frac{s\log^{3/4} N}{T}$ and define $\tilde{\lambda}_T = f_T^2\lambda_T$, such that the estimator can equivalently be defined as the minimizer of
\begin{equation}\label{eq:obj_fun}
    \mathcal{L}(\bbeta) = f_T^2\norm{\by - \bX\bbeta}_2^2 + \tilde{\lambda}_T\norm{\bbeta}_1.
\end{equation}
Next, define the set 
\begin{equation}\label{eq:empirical_process}
  \mathcal{A}_T\left(\tilde{\lambda}_2\right) = \left\lbrace f_T^2\norm{\bX^\prime\bepsilon_y}_\infty \leq \frac{\tilde{\lambda}_T}{4} \right\rbrace.
\end{equation}
By construction $\mathcal{L}\left(\hat{\bbeta}\right) \leq \mathcal{L}\left(\bbeta\right)$. Then, on the set $\mathcal{A}_T$, it holds that
\begin{equation}\label{eq:bound_step1}
    \begin{split}
    &f_T^2\norm{\by - \bX\hat{\bbeta}}_2^2 + \tilde{\lambda}_T\norm{\hat{\bbeta}}_1 \leq f_T^2\norm{\by - \bX\bbeta}_2^2 + \tilde{\lambda}_T\norm{\bbeta}_1\\
    &\begin{split}\Rightarrow f_T^2\norm{\bX(\hat{\bbeta} - \bbeta)}_2^2 + \tilde{\lambda}_T\norm{\hat{\bbeta}}_1  &\leq 2f_T^2\norm{\bX^\prime\bepsilon_y}_\infty\norm{\hat{\bbeta}- \bbeta}_1 + \tilde{\lambda}_T\norm{\bbeta_S}_1\\
    &\leq \frac{\tilde{\lambda}_T}{2}\norm{\hat{\bbeta}- \bbeta}_1 + \tilde{\lambda}_T\norm{\bbeta_S}_1\end{split}
    \end{split}.
\end{equation}
Using that 
\begin{equation*}
    \norm{\hat{\bbeta}}_1 = \norm{\hat{\bbeta}_S}_1 + \norm{\hat{\bbeta}_S{^c}}_1 \geq \norm{\hat{\bbeta}_S}_1 - \norm{\hat{\bbeta}_S - \bbeta_S}_1 + \norm{\hat{\bbeta}_{S^c}}_1,
\end{equation*}
it follows from \eqref{eq:bound_step1} that
\begin{equation}\label{eq:beta_cone}
\begin{split}
    2\tilde{\lambda}_T\norm{\hat{\bbeta}}_1 \leq \tilde{\lambda}_T\norm{\hat{\bbeta}- \bbeta}_1 + 2\tilde{\lambda}_T\norm{\bbeta_S}_1 \Rightarrow \norm{\hat{\bbeta}_{S^c}}_1 \leq 3\norm{\hat{\bbeta}_S - \bbeta_S}_1.
\end{split}
\end{equation}
Define the cone
\begin{equation}\label{eq:cone}
    \mathcal{C}_N(S,c_0) := \left\lbrace \bx \in \mathbb{R}^N : \norm{\bx_{S^c}}_1 \leq c_0\norm{\bx_S}_1 \right\rbrace
\end{equation}
and note that \eqref{eq:beta_cone} implies that $\hat{\bbeta} - \bbeta \in \mathcal{C}_N(S,3)$. Recall the definition of $\kappa\left(\hat{\bSigma}_T,s,c_0\right)$ in \eqref{eq:restricted_eigenvalue} and define the set
\begin{equation}\label{eq:set_cone}
    \mathcal{B}_T = \left( \kappa\left(\hat{\bSigma}_T,s,c_0\right) \geq \phi_0 \right\rbrace
\end{equation}
assume that $\kappa\left(\hat{\bSigma}_T,s,3\right) \geq \phi_0^2$. Then, combining \eqref{eq:bound_step1} and \eqref{eq:beta_cone},
\begin{equation}\label{eq:bound_step2}
\begin{split}
    &2\norm{f_T\bX(\hat{\bbeta}-\bbeta)}_2^2 + \tilde{\lambda}_T\norm{\hat{\bbeta} - \bbeta}_1 \overset{\eqref{eq:bound_step1}}{\leq} 4\tilde{\lambda}_T\norm{\hat{\bbeta}_S - \bbeta_S}_1 \leq 4\sqrt{s}\tilde{\lambda}_T\norm{\hat{\bbeta}_S - \bbeta_S}_2\\
    &\leq \frac{4\tilde{\lambda}_T\sqrt{s}\norm{f_T\bX\left(\hat{\bbeta}-\bbeta\right)}_2}{\phi_0} \leq \frac{4\tilde{\lambda}_T^2s}{\phi_0^2} + \norm{f_T\bX\left(\hat{\bbeta}-\bbeta\right)}_2^2,
\end{split}
\end{equation}
with the last inequality following from the elementary inequality $2uv \leq u^2 + v^2$. Hence, we are left to conclude that, on $\mathcal{A}_T\left(\tilde{\lambda}_T\right) \cap \mathcal{B}_T$,
\begin{equation*}
\begin{split}
    &\norm{f_T\bX(\hat{\bbeta}-\bbeta)}_2^2 + \tilde{\lambda}_T\norm{\hat{\bbeta} - \bbeta}_1 \leq \frac{4\tilde{\lambda}_T^2s}{\phi_0^2}\\
    &\Leftrightarrow \norm{\bX(\hat{\bbeta}-\bbeta)}_2^2 + \lambda_T\norm{\hat{\bbeta} - \bbeta}_1 \leq \frac{4f_T^2\lambda_T^2s}{\phi_0^2} = \frac{4\lambda_T^2s^3\log^{3/2} N}{T^2\phi_0^2}
\end{split}
\end{equation*}
Finally, we note that $\Prob\left(\mathcal{A}_T\left(\tilde{\lambda}_T\right) \cap \mathcal{B}_T\right) \leq 1 - \Prob\left(\mathcal{A}_T\left(\tilde{\lambda}_T\right)^c\right) - \Prob\left(\mathcal{B}_T^c\right)$. By Lemma \ref{Lemma:emp_proc}, it holds that
\begin{equation*}
\begin{split}
    &\Prob\left(\mathcal{A}_T^c\right) \leq 2N\left[\text{exp}\left(-\frac{2\tilde{\lambda}_T^2T^{3-2m}}{s^4\log^3 N}\right) + 2T\text{exp}\left(-\frac{T^{m-\frac{1}{2}}}{2C_\sigma}\right) + 2\text{exp}\left(-\frac{\tilde{\lambda}_T T^2}{4C_\sigma s^2\log^{3/2} N}\right)\right]\\
    &= 2N\left[\text{exp}\left(-\frac{2\lambda_T^2}{T^{1+2m}}\right) + 2T\text{exp}\left(-\frac{T^{m-\frac{1}{2}}}{2C_\sigma}\right) + 2\text{exp}\left(-\frac{\lambda_T}{4C_\sigma}\right)\right],
\end{split}
\end{equation*}
whereas it follows from Theorem \ref{Thm:REC} that
\begin{equation*}
    \Prob\left(\mathcal{B}_T^c\right) \leq (4T + C_2s)\text{e}^{-C_1s\log N}.
\end{equation*}
Combining the latter two inequalities completes the proof.
\end{proof}

\begin{proof}[\textbf{Proof of Corollary \ref{Cor:lasso_consistency}}]
We show that under the the rate restrictions in Corollary \ref{Cor:lasso_consistency}, and Assumption \ref{assump:error_distribution}, the error bound in Theorem \ref{Thm:error_bounds_lasso} holds with probability converging to 1 as $T,N,s \to \infty$. This requires us to show that
\begin{equation*}
\begin{split}
     &2N\left[\text{exp}\left(-\frac{2\lambda_T^2}{T^{1+2m}}\right) + 2T\text{exp}\left(-\frac{T^{m-\frac{1}{2}}}{2C_\sigma}\right) + 2\text{exp}\left(-\frac{\lambda_T}{4C_\sigma}\right)\right] +z (4T + C_2s)e^{-C_1s\log N} \\
     &= A_1 + A_2 + A_3 + A_4 \to 0,
\end{split}
\end{equation*}
as $T,N,s \to \infty$. For $A_1$, we set $m = 1/2+\xi$, and note that
\begin{equation*}
    A_1 = \text{exp}\left(-\frac{2\lambda_T^2}{T^{2+2\xi}} + \log N\right) \to 0
\end{equation*}
if $\frac{T^{1+\xi}\log^{1/2} N}{\lambda_T} \to 0$ as $T,N,s \to \infty$. Next, note that
\begin{equation*}
    A_2 = \text{exp}\left(-\frac{T^{\xi}}{2C_\sigma} + \log N + \log 2T\right) \to 0,
\end{equation*}
if $\frac{\log N}{T^\xi} \to 0$. Moving to $A_3$, it holds that
\begin{equation*}
    A_3 = 2\text{exp}\left(-\frac{\lambda_T}{4C_\sigma } + \log N\right) \to 0,
\end{equation*}
if $\frac{\log N}{\lambda_T } \to 0$ as $T,N,s \to \infty$, which is also implied by the combination of conditions that $\frac{T^{1+\xi}\log^{1/2} N}{\lambda_T} \to 0$ and $\frac{\log N}{T^\xi} \to 0$ as $T,N,s \to \infty$. Finally, it holds that
\begin{equation*}
    A_4 = (4T + C_2s)\text{e}^{-C_1s\log N} = 4\text{exp}\left(-C_1s\log N + \log T\right) + C_2\text{exp}\left(-C_1s\log N + \log s\right) \to 0,
\end{equation*}
if $\frac{s\log N}{\log T} \to \infty$ and $\frac{s\log N}{\log s} \to \infty$, as $T,N,s \to \infty$. The latter condition holds trivially, while the former condition is implied by the condition $\frac{\log T}{s \log N} \to 0$ as $T,N,s \to \infty$. As a result, it follows that $\mathcal{P}(T,N,s) \to 1$ as $T,N,s \to \infty$, with $\mathcal{P}(T,N,s)$ being defined in \eqref{eq:probability_error_bound}, thereby completing the proof.
\end{proof}
\end{small}
\end{appendices}

%-------------------------------------------------------------------
\newpage
\clearpage
\label{Bibliography}
\bibliographystyle{apalike}
\bibliography{Bibliography}

\end{document}